\documentclass[11pt]{article}
\usepackage[usenames,dvipsnames,svgnames,table]{xcolor}
\usepackage{enumitem}
\usepackage{graphicx}
\usepackage{fancybox}
\usepackage{comment}
\usepackage{xcolor}
\usepackage{nameref}
\usepackage{bbm}

\definecolor{ForestGreen}{rgb}{0.1333,0.5451,0.1333}
\definecolor{DarkRed}{rgb}{0.8,0,0}
\definecolor{Red}{rgb}{1,0,0}
\usepackage[linktocpage=true,
pagebackref=true,colorlinks,
linkcolor=ForestGreen,citecolor=ForestGreen,
bookmarks,bookmarksopen,bookmarksnumbered]
{hyperref}
\usepackage[nottoc,numbib]{tocbibind}

\usepackage{amsmath}
\usepackage{amssymb}
\usepackage{amsthm}

\usepackage[ruled, noend, linesnumbered]{algorithm2e}

\usepackage{subfig}

\usepackage{thm-restate}

\usepackage{float}
\usepackage{cleveref}

\newtheorem{theorem}{Theorem}[section]

\newtheorem{corollary}[theorem]{Corollary}
\newtheorem{lemma}[theorem]{Lemma}

\newtheorem{claim}[theorem]{Claim}

\newtheorem{invariant}[theorem]{Invariant}

\newtheorem{definition}[theorem]{Definition}
\newtheorem{remark}[theorem]{Remark}

\newtheorem*{theorem*}{Theorem}
\newtheorem*{corollary*}{Corollary}
\newtheorem*{conjecture*}{Conjecture}
\newtheorem*{lemma*}{Lemma}
\newtheorem*{thm*}{Theorem}
\newtheorem*{prop*}{Proposition}
\newtheorem*{obs*}{Observation}
\newtheorem*{definition*}{Definition}

\newtheorem*{remark*}{Remark}
\newtheorem*{rec*}{Recommendation}

\newenvironment{fminipage}%
  {\begin{Sbox}\begin{minipage}}%
  {\end{minipage}\end{Sbox}\fbox{\TheSbox}}

\def\defeq{\stackrel{\mathrm{def}}{=}}

\def\ceil#1{\left\lceil #1 \right\rceil}

\renewcommand{\deg}{\operatorname{deg}}

\renewcommand{\hat}{\widehat}
\renewcommand{\tilde}{\widetilde}

\DeclareFontFamily{U}{mathb}{\hyphenchar\font45}
\DeclareFontShape{U}{mathb}{m}{n}{<5> <6> <7> <8> <9> <10> gen * mathb
<10.95> mathb10 <12> <14.4> <17.28> <20.74> <24.88> mathb12}{}
\DeclareSymbolFont{mathb}{U}{mathb}{m}{n}
\DeclareMathSymbol{\rcirclearrow}{\mathbin}{mathb}{'367}

\newif\ifrandom
\randomtrue

\newcommand{\todolater}[1]{}

\DeclareUnicodeCharacter{2113}{$\ell$}

\usepackage{fullpage}

\newcommand{\simon}[1]{{\textbf{ \color{orange} Simon: #1}}}
\newcommand{\mprobst}[1]{{ \textbf{\color{purple} Max: #1}}}

\renewcommand{\simon}[1]{}
\renewcommand{\mprobst}[1]{}

\DeclareMathOperator{\vol}{\mathbf{vol}}

\newcommand*\samethanks[1][\value{footnote}]{\footnotemark[#1]}

\title{Dynamic Connectivity with Expected \\Polylogarithmic Worst-Case  Update Time}
\author{Simon Meierhans\thanks{The research leading to these results has received funding from grant no. 200021 204787 of the Swiss National Science Foundation. Simon Meierhans is supported by a Google PhD Fellowship.} \\
ETH Zurich \\
mesimon@inf.ethz.ch 
\and  
Maximilian Probst Gutenberg\samethanks[1] \\
ETH Zurich \\
maximilian.probst@inf.ethz.ch
}
\date{}

\begin{document}

\maketitle

\begin{abstract}

Whether a graph $G=(V,E)$ is connected is arguably its most fundamental property. Naturally, connectivity was the first characteristic studied for dynamic graphs, i.e. graphs that undergo edge insertions and deletions. While connectivity algorithms with polylogarithmic amortized update time have been known since the 90s, achieving worst-case guarantees has proven more elusive.

Two recent breakthroughs have made important progress on this question: (1) Kapron, King and Mountjoy [SODA'13; Best Paper] gave a Monte-Carlo algorithm with polylogarithmic worst-case update time, and (2) Nanongkai, Saranurak and Wulff-Nilsen [STOC'17, FOCS'17] obtained a Las-Vegas data structure, however, with subpolynomial worst-case update time. Their algorithm was subsequently de-randomized [FOCS'20].

In this article, we present a new dynamic connectivity algorithm based on the popular core graph framework that maintains a hierarchy interleaving vertex and edge sparsification. Previous dynamic implementations of the core graph framework required subpolynomial update time. In contrast, we show how to implement it for dynamic connectivity with \textbf{polylogarithmic expected worst-case update time}.

We further show that the algorithm can be de-randomized efficiently: a deterministic static algorithm for computing a connectivity edge-sparsifier of low congestion in time $T(m) \cdot m$ on an $m$-edge graph yields a deterministic dynamic connectivity algorithm with $\tilde{O}(T(m))$ worst-case update time. Via current state-of-the-art algorithms [STOC'24], we obtain $T(m) = m^{o(1)}$ and recover deterministic subpolynomial worst-case update time.
\end{abstract}

\section{Introduction}

The fully-dynamic connectivity problem is one of the most central problems on dynamic graphs. In its most basic form, it asks to maintain whether a dynamic graph is connected. This has been the first problem on dynamic graphs that has been studied, and continues to be researched extensively \cite{spira1975finding, chin1978algorithms, harel1982line, frederickson1985data, HenzingerKing95, holm01, thorup2000near, patrascu2007planning, kapron, WulffNilsen13,wang2015improvedrandomizeddatastructure, gibb2015dynamicgraphconnectivityimproved, NS17, W17, NSW17, chuzhoy2020deterministic, expander_hierarchy, theoretics:9645}.

Many techniques that later found broad application in (dynamic) graph algorithms were developed in the context of connectivity. These include reductions from general to sparse graphs \cite{frederickson1985data, it97}, topology-trees and top trees \cite{fred97, alstrup}, reductions from fully-dynamic to decremental graphs \cite{HenzingerKing95, holm01}, the emergency planning setting \cite{patrascu2007planning}, and most recently the development of expander pruning techniques \cite{W17, NS17, NSW17, SW19}. 

Although deterministic algorithms with amortized polylogarithmic update time were developed in the 90s \cite{HenzingerKing95, holm01}, turning amortized into worst-case bounds has proven much more challenging. To date, finding a deterministic algorithm with polylogarithmic worst-case update time remains a major open problem. 

\vspace{-10pt}\paragraph{The Kapron-King-Mountjoy algorithm.} In 2013, Kapron, King, and Mountjoy \cite{kapron} presented a breakthrough algorithm that maintains a spanning tree of a dynamic graph with worst-case polylogarithmic update time—an achievement that was awarded the SODA'13 best paper award. Their algorithm maintains a spanning tree, and whenever a tree edge gets deleted, it re-connects the tree with another edge crossing this cut as long as such an edge exists. If there is exactly one edge left in the cut after the deletion, the authors show that this edge can be found via a simple bit trick. Then, they show that adequately sub-sampling the graph reduces to this case. Although this step is algorithmically simple, the analysis is highly non-trivial.

This elegant algorithm has two major drawbacks: 
\begin{itemize}
    \item \textbf{Monte-Carlo property:} Sub-sampling only works with high probability, and if the algorithm fails to recover an edge, it is impossible to know if we got very unlucky or if the graph is truly disconnected. Thus, the algorithm is only correct with high probability.
    \item \textbf{Oblivious adversary assumption:} the algorithm only succeeds if updates are independent of the internal state of the algorithm. This means, in particular, that an adversary with access to the underlying spanning tree can foil the algorithm. 
\end{itemize}

\vspace{-15pt}\paragraph{The Nanongkai-Saranurak-Wulff-Nilsen Algorithm.} In 2017, Nanongkai, Saranurak and Wulff-Nilsen \cite{NSW17}, building on \cite{W17, NS17}, presented a Las-Vegas algorithm for maintaining a minimum spanning tree that was subsequently de-randomized in \cite{chuzhoy2020deterministic}. Their algorithm is based on expander graph contractions, and combines them with ideas that date back to \cite{frederickson1985data}. To dynamically maintain expander graphs, they developed expander pruning. The ensuing expander framework was highly influential in both static and dynamic algorithms, and expander pruning has among others led to breakthrough algorithms for computing expander decompositions \cite{SW19}, and is heavily used in the currently fastest and arguably simplest algorithm for solving decremental and static minimum-cost flow \cite{decrflow}. 

While this algorithm addresses the drawbacks of the \cite{kapron}-algorithm, it again has two major drawbacks:
\begin{itemize}
    \item \textbf{Subpolynomial Overhead:} The framework from \cite{NSW17} seems to inherently incur subpolynomial worst-case update times.  Even with optimal static subroutines, the algorithm cannot achieve polylogarithmic worst-case bounds.
    \item \textbf{Algorithmic Complexity:} The framework is technically intricate, and its de-randomization via \cite{chuzhoy2020deterministic} only further exacerbates this fact.
\end{itemize}
While the algorithmic complexity may not be inherent -- recent work \cite{expander_hierarchy} achieves similar results with greater simplicity and smaller subpolynomial factors -- the subpolynomial overhead appears fundamental to their contraction-based approach.

\vspace{-10pt}\paragraph{Our Contribution.} We present the first Las-Vegas algorithm for the dynamic connectivity problem with polylogarithmic worst-case time. 

\begin{theorem} \label{thm:simple}
    There is an algorithm that dynamically maintains a maximal spanning forest of a dynamic graph $G$ with polylogarithmic expected worst-case update time. The algorithmic guarantees hold against an adaptive adversary\footnote{In fact, our algorithm is non-oblivious, i.e., it works even against an adversary that has access to the entire state of the algorithm, including its random bits.} with access to the maximal spanning forest.\footnote{Our algorithm explicitly outputs changes to the spanning forest and ensures that it undergoes at most $2$ edge updates per update to $G$.}
\end{theorem}

Our algorithm is both conceptually simple and natural for dynamic connectivity.\footnote{We remark that any algorithm maintaining a spanning forest can be used to maintain a $(1 + \epsilon)$ approximate minimum spanning tree of a weighted graph $G = (V, E, w)$ with an additional multiplicative overhead of $\log(U)/\epsilon$ in the update time where $w(e) \in [1, U]$ (See \Cref{corr:apx_mst}).} The key technical contribution is implementing the highly influential dynamic core graph framework -- which also underlies state-of-the-art algorithms for maintaining distances \cite{maxflow, kyng2023dynamic, kyng2024bootstrapping}, electrical voltages \cite{chen2020fast}, and cuts \cite{decrflow} -- with only polylogarithmic overhead. While connectivity is arguably simpler than these problems, we believe our techniques can inspire progress on other dynamic graph problems.

Although our algorithm remains randomized, it outlines a clear path to the ultimate goal: a deterministic algorithm for fully-dynamic connectivity with polylogarithmic worst-case update time. We show that a fast deterministic algorithm for computing connectivity sparsifiers with low congestion yields such a result.

\begin{definition}[Graph Embedding]\label{def:graphEmbedding}
For graph $G= (V, E)$ and subgraph $H = (V, E') \subseteq G$, we say $\Pi_{G \mapsto H}$ is an \emph{embedding} if each edge $e = (u,v) \in E$ is mapped to a $uv$-path  $\Pi_{G \mapsto H}(e)$ in $H$. We define the \emph{congestion} of an edge $e' \in H$ to be the number of paths $\Pi_{G \mapsto H}(e)$ that contain $e'$,. The \emph{congestion} of $\Pi_{G \mapsto H}$ is defined to be the maximal congestion achieved by any $e' \in H$. 
\end{definition}

\begin{theorem}\label{thm:reduction}
Let $G=(V,E)$ be an $m$-edge graph. Given an algorithm $\mathcal{A}$ that given $G$ computes in deterministic time $T(m) \cdot m$ a subgraph $H \subseteq G$ with $\tilde{O}(n)$ edges along with an embedding $\Pi_{G \mapsto H}$ of congestion $T(m) \cdot \frac{m}{n}$. Then, there is a deterministic algorithm for the dynamic connectivity problem on a graph with $m$ edges with worst-case update time $\tilde{O}(T(m))$.
\end{theorem}

\Cref{thm:reduction} gives a clear path towards an efficient deterministic algorithm. While currently, there is not even a randomized algorithm achieving polylogarithmic $T(m)$, it is relatively straightforward to achieve subpolynomial $T(m)$ using expander embedding techniques implemented via efficient shortest path data structures (see \cite{maxflow, van2023deterministic, chen2023simple, kyng2023dynamic, kyng2024bootstrapping}). Thus, we retrieve a deterministic dynamic connectivity algorithm with subpolynomial worst-case update time.

We point out that the maximal spanning forest maintained by our algorithm only undergoes at most 2 edge updates per update to $G$ and thus changes can be outputted explicitly. 
Our algorithm also works against a non-oblivious adversary, i.e. an adversary that has access to the entire internal state of the algorithm including its random bits. 
Finally, we remark that by the black box reduction from \cite{bernstein2021deamortization}, the polylogarithmic expected worst-case update time could be transformed into an update time that is polylogarithmic with high probability.

\subsection{Overview} 

\paragraph{The Dynamic Core Graph Framework.} The dynamic core graph framework solves graph problems by maintaining a hierarchy that alternates between two types of sparsification: vertex-sparsifiers (called core graphs) and edge-sparsifiers.

Consider an input graph $G$ with $n$ vertices and $m = \tilde{O}(n)$ edges. Then, one cannot substantially reduce the number of edges in $G$ without compromising important graph properties. Instead, the framework selects a forest $F$ consisting of $n/\kappa$ components for some parameter $\kappa$, then contracts these components in $G$ to reduce the vertex count to $n/\kappa$. The graph after contraction is called a core graph and denoted by $\mathcal{C}(G, F)$. While $\mathcal{C}(G, F)$ does not necessarily reduce the number of edges, edge sparsification can now aim to reduce the edge count to $\tilde{O}(n/\kappa)$.

The central challenge in implementing this framework is finding a forest $F$ such that the core graph $\mathcal{C}(G, F)$ either preserves the desired property from $G$ or allows efficient extraction of this property. When dynamizing this approach, the time to update $\mathcal{C}(G, F)$ after an update to $G$ typically scales polynomially in $\kappa$. It is therefore desirable to choose $\kappa$ small and iterate the vertex and edge sparsification process until the problem size becomes small. This yields a hierarchy with $\Lambda = \Theta(\log_{\kappa}m)$ levels.

The key challenge in maintaining such a hierarchy dynamically is using core graphs and edge sparsifiers that can be maintained with small recourse. That is, for some small parameter $\tau$, each insertion or deletion to graph $G$ should cause at most $\tau$ changes in $\mathcal{C}(G, F)$ and its sparsifier. Since core graphs and their edge sparsifiers operate on top of each other, recourse multiplies across levels. Previous approaches obtained a total recourse bound of at most $\tau^{O(\Lambda)}$ across the hierarchy for each update to $G$.

Previous algorithms choose subpolynomial $\kappa$ (i.e., $\kappa = m^{o(1)}$) and typically achieve $\tau = \tilde{O}(1)$, resulting in subpolynomial overhead. In contrast, we need polylogarithmic $\kappa$ (i.e., $\kappa = \tilde{O}(1)$), which yields subpolynomial (or larger) recourse bounds \textbf{unless the recourse per level is at most one}, i.e., $\tau \leq 1$! This requires highly efficient algorithms combined with extremely careful analysis.

\paragraph{Formal Definition of the Dynamic Core Graph Framework.}
Given a graph $G = (V,E)$, its core graph with respect to a spanning forest $F \subseteq G$ is the graph $\mathcal{C}(G, F)$ obtained from $G$ by contracting the tree components in $F$ into single vertices and removing self-loops. An edge sparsifier $H \subseteq G$ is a graph that preserves some property of $G$ while being substantially sparser. Typically, for $G$ having $n$ vertices, one aims for $H$ to have only $\tilde{O}(n)$ edges.

The core graph framework maintains a hierarchy $\hat{G}_0, \hat{H}_0, \hat{G}_1, \hat{H}_1, \ldots, \hat{G}_{\Lambda}, \hat{H}_{\Lambda}$ (for $\Lambda = \Theta(\log_{\kappa} n)$). The graph $\hat{G}_0$ is the input graph $G$. For $i \geq 0$, $\hat{H}_i$ is an edge sparsifier of $\hat{G}_i$, and for $i > 0$, $\hat{G}_i = \mathcal{C}(\hat{H}_{i-1}, F_i)$ is the core graph of $\hat{H}_{i-1}$ obtained by contracting along the tree components of some forest $F_i$.

This completes the formal description of the core graph framework. Before proceeding, we note an important subtlety: deletion of an edge from a forest $F_i$ splits one of its components into two, resulting in a vertex split of the corresponding vertex in $\hat{G}_i$ (and $\hat{H}_i$). Thus, besides handling edge updates, one must also control recourse for such vertex splits.

\paragraph{Connectivity in the Core Graph Framework.}
While connectivity is arguably the simplest graph property to maintain, we believe this problem provides an excellent testbed for refining current techniques. We first outline a static version of our algorithm:
\begin{itemize}
    \item \textbf{Implementing the core graphs:} To form the core graphs, we find forests $F_i$ in the graphs $\hat{H}_{i-1}$ such that each component consists of at most $\kappa$ vertices.\footnote{This is slightly imprecise, we rather aim for each component to be incident to at most $\kappa$ volume with respect to $\hat{H}_{i-1}$.} In our algorithm, we choose $\kappa = \tilde{O}(1)$ since an edge deletion to $F_i$ might require investigating an entire component of $F_i$ to update $\hat{G}_i$. Since contractions preserve connectivity, the forest $F_i$ serves as a certificate for the connectivity of vertices contracted into super-vertices in $\hat{G}_i$.
    \item \textbf{Implementing the edge sparsification:} We first find a $\phi$-expander decomposition of $\hat{G}_i$, which is a vertex partition $X_1, X_2, \ldots, X_{\eta}$ such that:
      \begin{enumerate}
        \item for each $1 \leq \ell \leq \eta$, $\hat{G}_i[X_\ell]$ is a $\phi$-expander meaning for every cut $(S, X_{\ell} \setminus S)$, $|E(S, X_{\ell} \setminus S)| \geq \phi \cdot \min\{ \vol_{\hat{G}_i}(S), \vol_{\hat{G}_i}(X_{\ell} \setminus S)\}$, and
        \item the number of edges not in any such expander subgraph  $\hat{G}_i[X_\ell]$ is $q \cdot \phi \cdot m$ for some $q = \tilde{O}(1)$.
    \end{enumerate}
    We choose $\phi$ such that $q \cdot \phi \cdot m \leq m/4$, which can be achieved with $\phi = \tilde{\Omega}(1)$. Then, we let $\hat{H}_i$ be the graph obtained from these edges, and add for each expander graph  $\hat{G}_i[X_\ell]$ a spanning tree $T_{\ell}$. This effectively reduces the number of edges in $\hat{H}_i$ by a factor of $2$, i.e., $\hat{H}_{i}$ consists of less than half the number of edges in $\hat{H}_{i-1}$.

    Clearly, the connectivity on expanders is preserved in $\hat{H}_i$ by the spanning trees $T_{\ell}$, and the remaining connectivity information is simply added to $\hat{H}_i$. Thus, $\hat{H}_i$ preserves the connectivity information of $\hat{G}_i$.    
\end{itemize}
We interleave the above vertex sparsification and edge sparsification and, if $G$ is connected, obtain a hierarchy where $\hat{G}_{\Lambda} = \hat{H}_{\Lambda}$ consists of only polylogarithmically many vertices for some $\Lambda = O(\log m)$. Both core graphs and edge sparsification can be implemented in $\tilde{O}(m)$ randomized time via the classic techniques of Frederickson \cite{frederickson1985data} and the seminal result on expander decompositions from \cite{SW19}.

\paragraph{Controlling the Recourse in the Dynamic Algorithm.}
We next outline how to dynamize the algorithm above. As mentioned, the key challenge is controlling the recourse across the hierarchy. Recall that we cannot afford recourse per level larger than 1, as this would result in unacceptable blow-up over the levels. However, since edge insertions do not affect the connectivity certificates (i.e., the forests) of higher levels in the hierarchy, we can essentially ignore insertions and add them only to the top level graph $\hat{G}_{\Lambda}$.

On the other hand, an edge deletion might indeed require us to restore some connectivity information. There are three possible scenarios for how the deletion of an edge $e$ can affect a level $i$ in the hierarchy:
\begin{enumerate}
    \item if $e$ is in $F_i$, the deletion results in a vertex split of a vertex in $\hat{G}_i$ (and $\hat{H}_i$).
    \item if $e$ is in $\hat{H}_i$, then it can be processed as an edge deletion to our edge sparsifier.
    \item if $e$ is in neither, then we can ignore the effect of removing $e$ as no certificate at a higher level requires it.
\end{enumerate}

We show that in both of the first two cases, we can compensate for the deletion of $e$ by adding only $\tilde{O}(1)$ new edges to $\hat{H}_i$. Note that adding edges can be treated like edge insertions and thus effectively ignored and forwarded to the highest level of the hierarchy.

Of course, once we have added many edges to $\hat{G}_{\Lambda}$ (and $\hat{H}_\Lambda$), these graphs are no longer of polylogarithmic size. This can be mitigated by rebuilding graphs at level $i$ and higher every $2^i$ steps and factoring in edges that have been inserted since the last rebuild. Our approach achieves $\tilde{O}(1)$ worst-case update time under the assumption that edge deletions result at each level in adding at most $\tilde{O}(1)$ edges to the sparsifier $\hat{H}_i$. To do so, we distribute the computation of a rebuild over sufficiently many updates.  

\paragraph{Dynamic Edge Sparsification.} It remains to describe how to dynamize our edge sparsification procedure. Recall that we only have to react to vertex splits/edge deletions.

Let us first discuss how to deal with the more standard update of an edge deletion: when an edge $e$ is deleted from an expander graph $\hat{G}[X_\ell]$, we apply \emph{expander pruning} techniques to prune out a set $P_\ell \subseteq X_\ell$ of volume $\tilde{O}(1)$ such that $\hat{G}_i[X_\ell \setminus P_\ell]$ is still a $\tilde{\Omega}(1)$-expander. We rely on the recent new expander pruning algorithm from \cite{meierhans2025expanderpruningpolylogarithmicworstcase} that allows us to process many edge deletions to the expander while growing $P_i$ only $\tilde{O}(1)$ in volume at each step.\footnote{The seminal work in \cite{SW19} achieves this guarantee only via amortization; some steps might not grow $P_i$ at all, some grow it significantly. It is crucial for our algorithm to have worst-case guarantees. } 

Since the added volume to $P_\ell$ is in $\tilde{O}(1)$, we can afford to simply add all edges incident to $P_\ell$ to $\hat{H}_i$ (only one of the expander graphs can contain the deleted edge $e$). We are almost done restoring the invariants from initialization. However, we still have to deal with the case that $T_{\ell}$ might be no longer spanning $X_\ell \setminus P_\ell$ but instead only be a forest (for technical reasons, we delete edges incident to $P_\ell$ from $T_\ell$, however, crucially, we do not remove them from $\hat{H}_i$). In this case, let $C$ be the smallest connected component of the forest, then we have from the expander property that a $\tilde{\Omega}(1)$-fraction of the edges incident to $C$ leave $C$. We use random sampling to find such an edge and add it to $\hat{T}_i$ and $\hat{H}_i$. We then repeat this procedure until $T_{\ell}$ is spanning again.

Finally, the case of a vertex split can be handled similarly: by carefully constructing our hierarchy, we can ensure that each vertex in $\hat{G}_i$ is incident to at most $\tilde{O}(1)$ edges. We then simulate a vertex split by removing all incident edges to the vertex and then re-inserting the two resulting vertices with all their edges. This process can be simulated with an additional blow-up by a $\tilde{O}(1)$-factor in the recourse in terms of insertions.

\paragraph{Putting It All Together.} Let us briefly review the algorithm: we have given an implementation of the dynamic core graph framework where we obtain extremely tight control over the recourse by only passing edge deletions up the hierarchy to be processed and by collecting $\tilde{O}(1)$ new edges to be inserted to potentially recover connectivity and adding them only to the top level graph. To keep graphs small across the hierarchy, we use periodic rebuilds.

In our discussion, we have omitted many details to highlight the main conceptual ideas. The reader is referred to \Cref{sec:mainSec} for rigorous definitions and proofs.

\paragraph{De-randomization of the dynamic connectivity data structure.} Finally, let us briefly sketch how to obtain a deterministic data structure by following the same framework. Assume that edge sparsifier $\hat{H}_i$ can be computed along with an embedding $\Pi_{\hat{G}_i \mapsto \hat{H}_i}$ from $\hat{G}_i$ into $\hat{H}_i$ with congestion  $\alpha$ (see \Cref{def:graphEmbedding}). Consider the following dynamic update strategy: whenever an edge $e$ is deleted from $\hat{H}_i$, find all edges $e'$ from $\hat{G}_i$ with $e \in \Pi_{\hat{G}_i \mapsto \hat{H}_i}(e')$ and simply add them to $\hat{H}_i$. Clearly, the algorithm only adds at most $\alpha$ edges at each level. And, clearly, $\hat{H}_i$ preserves the connectivity information of $\hat{G}_i$ at all times: for every edge $e'$ in $\hat{G}_i$, either there still is an explicit path between its endpoints in $\hat{H}_i$, or the edge $e'$ is itself in $\hat{H}_i$. For $\alpha = \tilde{O}(1)$ and assuming such a sparsifier $\hat{H}_i$ along with an embedding can be computed in deterministic near-linear time, this yields a deterministic dynamic connectivity algorithm with polylogarithmic worst-case update time.

We give a more detailed discussion of the reduction and other approaches towards efficient de-randomization in \Cref{sec:derandomziation}.

\section{Preliminaries}

\paragraph{Misc. } We let $[n] = 1, \ldots, n$, and $[n]_0 = 0, \ldots, n$. 

\paragraph{Graphs. } We let $G = (V, E)$ denote a graph with vertex set $V$ and edge set $E \subseteq V \times V$. For $v \in V$, we let $\deg_G(v)$ denote the number of edges adjacent to vertex $v$ and for any $A \subseteq V$, we let $\vol_G(A) = \sum_{a \in A} \deg_G(a)$ denote the volume of $A$.

For $A,B \subseteq V$, we let $E(A, B) \defeq \{(u,v) \in E: u \in A \wedge v \in B\}$ be the set of edges with one endpoint in $A$ and one endpoint in $B$. For $A \subset V$, we let $G[A]$ denote the induced subgraph on $A$, i.e. $G[A] = (A, E(A,A))$. For notational convenience, we sometimes write $G \cup E'$ instead of $(V,E \cup E')$ and $G \setminus E'$ instead of $(V, E \setminus E')$. 

In a graph $G = (V, E)$, we call a vertex set $\emptyset \subset A \subseteq V$ connected (in $G$) if $E(B, A \setminus B) \neq \emptyset$ for all $B$ such that $\emptyset \subsetneq B \subsetneq A$.  If additionally $E(A, V \setminus A) = \emptyset$, the we refer to $A$ as a connected component of $G$. We call the partition of the vertex set into such components the connected components of $G$.

\paragraph{Trees and Forests. } We call a graph $G = (V, E)$ a tree if it is connected and has exactly $|V| - 1$ edges. We call a collections of trees a forest. We usually use $T$ and $F$ to refer to graphs that are trees and forests respectively.

We call a forest $F$ a \emph{spanning forest} of $G$ if they share the same vertex set and $F \subseteq G$ and say it is a \emph{maximal spanning forest} if $F$ and $G$ additionally share the same connected components.

\paragraph{Dynamic Graphs. } A dynamic graph $G$ is a graph whose vertex and edge sets change over time. Formally, we define such a dynamic graph as a tuple of sequences $((G^{(i)})_{i \in [T]_0}, (U^{(i)})_{i \in [T]})$. Then, the entry $G^{(0)}$ refers to the intitial graph, and $G^{(i)}$ is obtained by applying the update sequence $U^{(i)}$ to $G^{(i - 1)}$. An update sequence $U^{(i)}$ may contain updates of the following types. 
\begin{enumerate}
    \item \underline{Edge Insertion:} Adds a new edge to the edge set of the graph. \label{update:insert} 
    \item \underline{Edge Deletion:} Deletes an edge from the edge set of the graph. \label{update:delete} 
    \item \underline{Vertex Split:} Replaces a vertex $v$ in the graph with two new vertices $v'$ and $v''$. Every edge incident to $v$ becomes incident to either $v'$ or $v''$. \label{update:split} 
\end{enumerate}

\section{Dynamic Maximal Spanning Forest via Core Graphs}
\label{sec:mainSec}
In this section, we describe how to maintain a maximal spanning forest of a graph with only polylogarithmic update time.

\begin{restatable}{theorem}{mainThm}\label{thm:mainTech}
Assume to be given an edge-dynamic graph $G=(V,E)$, initially consisting of $m$ edges undergoing a sequence of at most $m$ edge insertions/deletions. Then, there is a data structure with initialization time $\tilde{O}(m)$ that explicitly maintains an edge-dynamic forest $F \subseteq G$ at any step with worst-case update time $\tilde{O}(1)$ (holding w.h.p.), such that $F$ remains a maximal spanning forest of $G$ at all times.
\end{restatable}

Note that \Cref{thm:simple} claims expected worst case update time instead of polylogarithmic worst-case update time. These guarantees are equivalent up to polylogarithmic loss in efficiency by a black box reduction from \cite{bernstein2021deamortization}.

By standard reductions (See \cite{frederickson1985data}), we assume w.l.o.g. that the input graph $G=(V,E)$ has maximum degree $3$ at all times. To further simplify the exposition, we first describe a version of our algorithm that achieves the guarantees in \Cref{thm:mainTech} with \emph{amortized} update time $\tilde{O}(1)$. We discuss a de-amortization of our algorithm via standard techniques in \Cref{sec:deamortization}.

\subsection{The Main Algorithm} 

Our algorithm maintains a hierarchy of core graphs. We start by giving the definition.

\begin{definition}[Core Graph] \label{def:core_graph}
    Given a graph $G = (V, E)$ and a spanning forest $F$ we let $\mathcal{C}(G, F)$ denote the graph obtained from $G$ via contracting the tree components in $F$ into single vertices and removing self-loops. 
\end{definition}
\begin{remark}
When $\mathcal{C}(G, F)$ is clear from context, we denote for each vertex $u \in V$, by $\hat{u}$ the image of $u$ under the contraction operation. For each edge $e = (u,v) \in E$, we denote by $\hat{e}$ either the edge $(\hat{u}, \hat{v})$ for $\hat{u} \neq \hat{v}$ or an empty pointer $\bot$ for $\hat{u} = \hat{v}$. We also say that $u$ is in the pre-image of $\hat{u}$ and $e$ is in the pre-image of $\hat{e}$.
\end{remark}

Observe that for core graph $\mathcal{C}(G, F)$ and $F \subseteq G$, we have that for every $u,v\in V$, $u$ and $v$ are connected in $G$ iff $\hat{u}$ and $\hat{v}$ are connected in $\mathcal{C}(G, F)$. This directly follows because the contracted components are connected, which is certified by the corresponding tree. We exploit this fact (often implicitly) throughout.

As in previous dynamic algorithms that employed core graphs, the key to efficiency of our data structure is that whenever we take a core graph, the underlying forest $F$ has each tree $T$ in $F$ incident to small volume. We make this constraint qualitative and say $F$ is $\kappa$-shattering if the volume incident to each tree is at most $\kappa$.

\begin{definition}[Shattering Forest] \label{def:shattering_forest}
Given graphs $G = (V, E)$ and $H \subseteq G$, we call a forest $F \subseteq G$ consisting of trees $T_1, \ldots, T_\eta$ a $\kappa$-shattering forest with respect to $H$ if $vol_H(V(T_i)) \leq \kappa$ for all $i \in 1, \ldots, \eta$ and $V(G) = V(F)$. 
\end{definition}

Our data structure maintains a hierarchy over $\Lambda + 1$ layers for $\Lambda = \ceil{\log_2 m} + 4$.
At initialization and after every update, we update the hierarchy one layer after another in increasing order of $i$, i.e. we process layer $i$ before processing layer $i+1$. We let time $t$ refer to the number of updates to $G$ processed, that is, initialization takes place at time $0$, and the processing after the $t$-th update to $t$ takes place at time $t$. We define for $0 \leq i \leq \Lambda$, $x_i \defeq \ceil{2^{\Lambda - i - 3}}$. Then, we \emph{re-initialize} layer $i$ at every time step $t$ divisible by $x_i$.

We let $F_0 = (V, \emptyset), \hat{G}_0 = G, \hat{H}_0 = G, H_0 = G$ at all times. At time $t$, at each layer $i = 1, \ldots, \Lambda$, the data structure maintains:
\begin{enumerate}
    \item \underline{$F_{i}$:} A spanning forest with $F_i \subseteq F_{i-1} \cup H_{i-1}$. We distinguish as follows:
    \begin{itemize}
        \item \underline{If $t$ is divisible by $x_i$:} $F_i$ is re-initialized to be $\kappa$-shattering forest w.r.t. $H_{i-1}$ where $\kappa \defeq 8 \cdot 108$. 
        \item \underline{Otherwise:} if at time $t$ an edge $e$ is deleted from $G$, we also remove $e$ from $F_i$ (to ensure $F_i \subseteq F_{i-1} \cup H_{i-1}$). Otherwise, $F_i$ is not updated.
    \end{itemize}
    
    \item \underline{$\hat{G}_i$:} The current core graph $\mathcal{C}(F_i \cup H_{i-1}, F_i)$. 
    
    \item \underline{$\hat{H}_i$:} A \emph{connectivity sparsifier} of $\hat{G}_i$ at all times. That is, $\hat{H}_i \subseteq \hat{G}_i$ and $u, v \in V_{\hat{H}_i}$ are connected $\hat{H}_i$ iff they are connected in $\hat{G}_i$. Again, we distinguish as follows:
    \begin{itemize}
        \item \underline{If $t$ is divisible by $x_i$:} we rebuild $\hat{H}_i$ to be of size at most $O(|V_{\hat{H}_i}|)$.
        \item \underline{Otherwise:} after every update to $G$, $H_{i-1} \cup F_{i}$ is updated before updating $\hat{H}_i$. For every insertion of edge $e =(u,v)$ to $H_{i-1} \cup F_{i}$, we add the edge $\hat{e} = (\hat{u}, \hat{v})$ to $\hat{H}_i$ if $\hat{u} \neq \hat{v}$ (w.r.t. $\hat{G}_i$). Further, we will show that there is at most one deletion to $H_{i-1} \cup F_{i}$ per update to $G$ at a later stage. 
        
        If such an edge deletion occurs, say edge $e = (u,v)$ is deleted from $H_{i-1} \cup F_{i}$. Let $\hat{u}^{\textit{OLD}}, \hat{v}^{\textit{OLD}}$ refer to $u, v$ in the graph $\hat{G}_i$ before processing at time $t$, and $\hat{u}^{\textit{NEW}}, \hat{v}^{\textit{NEW}}$ after. In particular, if $e \in F_i$, then $\hat{u}^{\textit{OLD}} = \hat{v}^{\textit{OLD}}$. 

        Then, we simlate that $\hat{u}^{\textit{OLD}}$ and $\hat{v}^{\textit{OLD}}$ are deleted from $\hat{H}_i$, identify a set of edges $\Delta \hat{E}_i$ that ensures that upon being added, $\hat{H}_i \setminus \{\hat{u}, \hat{v}\} \cup \Delta \hat{E}_i$ is a connectivity sparsifier of $\hat{G}_i \setminus \{\hat{u}^{\textit{NEW}}, \hat{v}^{\textit{NEW}}\}$ and finally add to $\hat{H}_i$ the vertices $\hat{u}^{\textit{NEW}}, \hat{v}^{\textit{NEW}}$ along with all their edges incident in $\hat{G}_i$. \textbf{$\hat{u} \mapsto \hat{v}^{OLD}?$}
    \end{itemize}
    \item \underline{$H_i$:} The pre-image of $\hat{H}_i$, that is $H_i = (V, \{e \;|\; \hat{e} \in E(\hat{H}_i)\})$.
\end{enumerate}

We update each of these components in the order that they are listed. We observe that every forest $F_i \subseteq G$. Further, our algorithm maintains the following key invariant for the hierarchy, whose proof is deferred to \Cref{sec:analysis_main}.

\begin{invariant}\label{inv:keyInv}
For any layer $i \in [\Lambda]$, we have 
\begin{equation*}
    \text{(\#connected components in $F_i$)} \leq \text{(\#connected components in $G$)} + 2^{\Lambda - i - 1}.
\end{equation*}
\end{invariant}

We let $F_{\Lambda}$ be the spanning forest that our algorithm outputs. By \Cref{inv:keyInv} and $F_{\Lambda} \subseteq G$, $F_{\Lambda}$ is a maximal spanning forest of $G$. 

We store each forest $F_i$ in a link-cut tree data structure with worst-case update time $O(\log^2 n)$ \cite{driscoll1986making}. This link-cut tree data structure additionally guarantees that a copy of the whole data structure can be made in worst-case time $O(\log^2 n)$.\footnote{This data structure falls into the more general category of \emph{persistent} data structures \cite{driscoll1986making}. The copying is then emulated by storing differences to the moment that data structure was copied.} 

We give the missing implementation details in the sections below and then analyze the main data structure, which yields \Cref{thm:mainTech}.

\subsection{Implementing the Re-initialization of $F_i$} \label{sec:forest_reinit}
\textbf{Below, $F_[v,x] \mapsto F_{i-1}[v,x]$}
\paragraph{The re-initialization procedure.} Recall that $F_i$ is re-initialized to be a $\kappa$-shattering forest w.r.t. $H_{i-1}$. Let us now describe the precise procedure: We let $E_i$ be the set of edges in $F_{i-1}$ that are incident to an endpoint of an edge in $E_{H_{i-1}}$. The algorithm first computes the branch free set $B_i$ of $V_{E_i}$ in $F_{i-1}$, i.e. $B_i$ is the set of all vertices in $V_{E_i}$ and all vertices $v \in V$ for which there exist three vertices $x,y,z \in V_{E_i}$ in the component of $v$ such that $F_{i - 1}[v,x], F_{i - 1}[v,y], F_{i - 1}[v, z]$ are pairwise disjoint. Then, it constructs the graph $\hat{A}_i = \mathcal{C}(F_{i-1} \cup H_{i-1}, F_{i-1} \setminus B_i)$ where $F_{i-1} \setminus B_i$ denotes the forest $F_{i-1}$ where all edges incident to $B_i$ are removed (See \Cref{fig:core_graph}). The reason for not directly computing a shattering forest of $\hat{H}_{i - 1}$ is that it does not have sufficiently bounded degree. 

\begin{figure}[ht]
    \centering
    \includegraphics[width=0.5\linewidth]{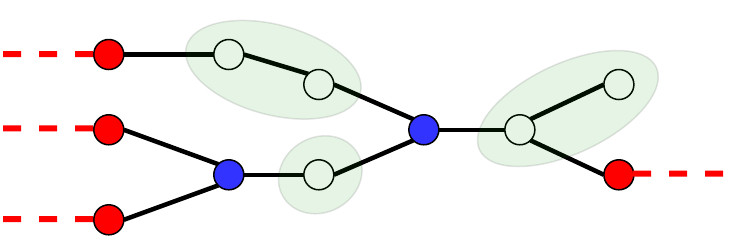}
    \caption{The black edges form a tree of $F_{i - 1}$, and the dashed red edges are in $H_{i - 1}$. Then, the (red) endpoints of the edges in $H_{i - 1}$ are included in the branch free set $B_i$. Finally, the blue vertices are added to $B_i$, such that the paths between adjacent vertices in $B_i$ become disjoint. The green shaded parts are contracted to obtain graph $\hat{A}_i$. The graph inherits maximum degree $3$ from $G$.}
    \label{fig:core_graph}
\end{figure}

We then take an arbitrary maximal spanning forest $\hat{F}_i$ of $\hat{A}_i$. For each component $\hat{T}$ in $\hat{F}_i$ of size at least $\kappa$, we invoke the algorithm in \Cref{lma:clusteringFrederikson} below with parameter $\kappa/9$. We let $\hat{E}_i$ be the set of all edges returned by these procedures. 

\begin{lemma}[see \cite{frederickson1985data}, Lemma 1]\label{lma:clusteringFrederikson}
For an $n$-vertex connected tree $T$ with maximum degree $3$, and a positive integer $z$ such that $n \geq 3z-2$, there is an algorithm that runs in time $O(n)$ and returns a set $E' \subseteq E(T)$ such that each connected component in $T \setminus E'$ contains at least $z$ and at most $3z-2$ vertices.
\end{lemma}

For components of size less than $\kappa$, we compute an arbitrary spanning tree. Finally, we let $F_i$ consist of the union formed by the pre-image of $\hat{F}_i \setminus \hat{E}_i$ and forest $F_{i-1} \setminus B$. 

\paragraph{Analysis.} Let us briefly analyze the re-initialization procedure.

\begin{claim}\label{clm:correctnessKappaShattering}
After re-initialization, $F_i$ is $\kappa$-shattering with respect to $H_{i-1}$. For every connected component $C$ of $G$, the forest $F_i[C]$ is either a tree, or it consists of at most $\frac{108}{\kappa} \cdot |E_{H_{i-1}[C]}|$ components. Finally, $|E_{\hat{A}_i}| \leq 5 \cdot |E_{H_{i-1}}|$. 
\end{claim}
\begin{proof}
By construction, we have for each $u \in B_i$ that $\hat{u}$ in core graph $\hat{A}_i$ have $u = \hat{u}$. Since $V(E_{H_{i-1}}) \subseteq B_i$, we have that every super-vertex $\hat{u}$ in $\hat{A}_i$ stemming from contraction of a non-singleton set $U$ is spanned by some tree $T_U \subseteq F_{i-1}$ that is not incident to \emph{any} edge in $H_{i-1}$.

We next prove that $A_i$ has maximum degree $3$ and thus the assumption of \Cref{lma:clusteringFrederikson} is satisfied for $\hat{F}_i \subseteq \hat{A}_i$: we first note that vertices $\hat{u}$ with $\hat{u} = u$ for some $u \in V$ have maximum degree $3$ since $F_{i-1} \cup H_{i-1} \subseteq G$ and $G$ having maximum degree $3$. Further observe that each $\hat{u}$ stemming from the contraction of a non-singleton set $U$ has at most $2$ edges leaving in $F_{i-1}$. Otherwise, we can find for each subtree $T_1, T_2, \ldots, T_k$ for $k \geq 3$ dangling from these edges at least one vertex $b_1, b_2, \ldots, b_k$ from $B_i$ in the respective subtree which implies that for some vertex $u \in U$, the intersection of paths $F_{i - 1}[b_1, b_2], F_{i - 1}[b_2, b_3], F_{i - 1}[b_1, b_3]$ is $u$ and thus $u$ is a branching vertex. But since $B_i \cap U$ is empty, this yields a contradiction. 

Next, fix any connected component $C$ of $G$. We use the well-known fact that the branch free set $B_i \cap C$ of set $|V_{E_{H_{i-1}[C]}}|$ is at most twice as large as the underlying set, i.e. $|B_i \cap C| \leq 2 \cdot |V_{E_{H_{i-1}[C]}}| \leq 4 \cdot |E_{H_{i-1[C]}}|$. It follows that, for $\hat{C}$ being the connected component $C$ after applying the same contractions that produced $\mathcal{A}_i$, we have $|V_{\mathcal{A}_i[\hat{C}]}| \leq 12 \cdot |E_{H_{i-1}[C]}|$ since there are no edges between vertices $\hat{u}$ stemming from the contraction of non-singleton sets and the vertices $\hat{u}$ stemming from singleton sets have maximum degree $3$. This yields, by \Cref{lma:clusteringFrederikson}, that $\hat{F}_i[C] \setminus \hat{E}_i[C]$ is $\kappa$-shattering w.r.t. the graph $\hat{A}_i$ (each cluster has at most $3\kappa/9 -2 < \kappa/3$ vertices, each of degree at most $3$) and that there are at most $|V_{\mathcal{A}_i}|/ (\kappa/9) \leq \frac{108}{\kappa} \cdot |E_{H_i}|$ connected components.

It then remains that 'uncontracting' the vertices $\hat{u}$ stemming from non-singleton sets $U$ in $\hat{F}_i[C] \setminus \hat{E}_i[C]$ such that the trees $T_U$ spanning $U$ are added again, does not connect any vertices in $B_i \cap C$, but also does not create additional components. Thus, $F_i[C]$ is $\kappa$-shattering w.r.t. $H_{i-1}[C] \subseteq F_{i-1}[C] \cup H_{i-1}[C]$ and the number of components is again at most $\frac{108}{\kappa} \cdot |E_{H_i[C]}|$.

Finally, the number of edges from $F_{i - 1}$ that are in $\hat{A}_i$ can be bounded because contracting the paths between vertices in the branch free sets yields trees with at most $2 \cdot |E_{H_{i-1}}|$ vertices, and therefore also edges. For every such contracted path, we add at most $2$ edges to $\hat{A}_i$. Therefore, the total number of edges in $\hat{A}_i$ is bounded by $5 \cdot |E_{H_{i - 1}}|$ as desired. 
\end{proof}

\begin{claim} \label{clm:forsest_rt}
The algorithm to re-initialize forest $F_i$ can be implemented in time $\tilde{O}(|E_{H_{i - 1}}|)$. 
\end{claim}
\begin{proof}
Let us follow the insights used in \Cref{clm:correctnessKappaShattering}. We can find the set $B_i$ of branching vertices in time $O(|B_i| \log n)$ by slightly augmenting the link-cut tree data structure on $F_{i-1}$ (for details, see \cite{decrflow}). Thus, we can construct the graph $\hat{A}_i$ explicitly in time $\tilde{O}(|E_{H_{i-1}}||)$. Computing a maximal spanning forest of $\hat{F}_i$ statically takes time $\tilde{O}(|E_{H_{i-1}}|)$ as does the procedure from \Cref{lma:clusteringFrederikson} to compute edge set $\hat{E}_i$. We can then construct $F_i$ in time $\tilde{O}(|E_{H_{i-1}}|)$ by making a copy of $F_{i-1}$ (using persistence), removing all edges incident to vertices in $B_i$, and adding the pre-images of all edges in $\hat{F}_{i} \setminus \hat{E}_i$. 
\end{proof}

\subsection{Implementing the Connectivity Sparsifier \texorpdfstring{$\hat{H}_i$}{TEXT} of \texorpdfstring{$\hat{G}_i$}{TEXT}}

In this section, we describe how to maintain connectivity sparsifiers $\hat{H}_i$ of $\hat{G}_i$ with the required guarantees.

\paragraph{Expanders and Expander Decompositions.} Expanders have been crucial in many previous connectivity data structures \cite{patrascu2007planning, NS17, W17, NSW17} and we also heavily employ them in our construction. We state a seminal result in the area.

\begin{definition}[Expander Decomposition]
For a graph $G$ and some set $A \subseteq V$, we let $\vol_G(A) = \sum_{a \in A} \deg_G(a)$ denote the volume of $A$. We then call $G$ a $\phi$-expander if $\frac{|E(A, V \setminus A)|}{\vol_G(A)} \geq \phi$ for every set $A \supset \emptyset$ such that $\vol_G(A) \leq \vol_G(V)/2$.
 
We say a partition $X_1, \ldots, X_{\eta}$ of the vertex set $V$ is a $\phi$-expander decomposition of quality $q$ if
    \begin{enumerate}
        \item $|\{(u,v) \in E(X_i, X_j): i \neq j\}| \leq q \cdot \phi \cdot m$ and
        \item $G[X_i]$ is a $\phi$-expander for all $1 \leq i \leq \eta$.
    \end{enumerate}
\end{definition}

\begin{theorem}[Expander Decomposition, See Theorem 1.2 in \cite{SW19}] \label{thm:expander_decomp}
    Given a graph $G = (V, E)$ with $m$ edges and $n$ vertices, and a parameter $\phi$, there is a randomized algorithm that computes a $\phi$-expander decomposition of quality $\gamma_{exp} = O(\log^3 m)$ in time $O\left(\frac{m \log^4 m}{\phi}\right)$. The algorithm succeeds w.h.p. 
\end{theorem}

It is well-known that expanders are relatively robust to updates, meaning that most of the expander remains well-connected, and efficient algorithms exist to uncover these well-connected parts by \emph{pruning} small parts that are potentially poorly connected. Here, we use a very recent result that achieves $\tilde{O}(1)$ \emph{worst-case} update times and recourse.

\begin{theorem}[Worst-Case Pruning, See Theorem 1.3 in \cite{meierhans2025expanderpruningpolylogarithmicworstcase}] \label{thm:worst_case_pruning}
    Given a $\phi$-expander $G = (V, E)$ with $n$ vertices and $m$ edges undergoing up to $\frac{\phi m}{\gamma_{\textit{prune}}}$ edge deletions for $\gamma_{\textit{prune}} = \tilde{O}(1)$, there is an algorithm that explicitly maintains a set $A$ such that $G[V \setminus A]$ remains a $\frac{\phi}{\gamma_{\textit{prune}}}$- expander. Every deletions/vertex split is processed in time $\gamma_{\textit{prune}}/\phi^2$ and the set $A$ grows by at most $\gamma_{\textit{prune}}/\phi^2$ vertices after each update. 
\end{theorem}

\paragraph{The Connectivity Sparsifier.} Equipped with these tools we are now ready to describe the algorithm to maintain $\hat{H}_i$:
\begin{itemize}
    \item \underline{If $t$ is divisible by $x_i$:} We re-compute $\hat{H}_i$ as follows. We first compute a $\phi_{\textit{sparse}}$-expander decomposition $X_1, X_2, \ldots, X_{\eta}$ of the current graph $\hat{G}_i$ for  $\phi_{\textit{sparse}} = 1/(20 \gamma_{\textit{exp}}) = \tilde{\Theta}(1)$ where $\gamma_{\textit{exp}}$ is the quality parameter in \Cref{thm:expander_decomp}. It then initializes a data structure $\mathcal{P}_j$ on each $\phi_{\textit{sparse}}$-expander $\hat{D}_j = \hat{G}_i[X_j]$. 

    Finally, we compute a spanning tree $\hat{T}_j$ for each expander $\hat{D}_j$ and re-initialize $\hat{H}_i$ as the union of all trees $\hat{T}_j$ and the set of edges crossing between expanders, i.e. the set $\{(u,v) \in E_{\hat{G}_i}(X_j, X_k) \;|\; j \neq k\}$.\footnote{If the graph $\hat{D}_j$ is disconnected, we restart the entire data structure. This only happens when the expander decomposition algorithm from \Cref{thm:expander_decomp} fails. }

    \item \underline{Otherwise:} As already mentioned previously, we will show that at every time step $t$, $H_{i-1} \cup F_i$ undergoes at most one edge deletion. 

    Let us now describe how to update $\hat{H}_i$ such that it is again a connectivity sparsifier of $\hat{G}_i$. We first handle the edge deletion to $H_{i-1} \cup F_i$ (if it occurs):  
    \begin{itemize}
        \item Let $e = (u,v)$ be the edge that is deleted. 
        For each $\hat{z} \in \{\hat{u}^{\textit{OLD}}, \hat{v}^{\textit{OLD}}\}$ (for $e \in F_i$, $\hat{u}^{\textit{OLD}} = \hat{v}^{\textit{OLD}}$), we do the following: we first remove all edges incident to $\hat{z}$ from $\hat{H}_i$.

        If there is a graph $\hat{D}_j$ as described above that contains $\hat{z}$, then we also remove all incident edges to $\hat{z}$ from $\hat{D}_j$ and inform the pruning data structure $\mathcal{P}_j$. If no such graph $\hat{D}_j$ exists, we do nothing in this step. Otherwise, if the data structure $\mathcal{P}_j$ received the maximum number of updates already, we add $\hat{D}_j$ to $\hat{H}_i$ and then set $\hat{D}_j$ to be the empty graph. 
        
        Otherwise, let $\Delta A$ be the set of vertices outputted by $\mathcal{P}_j$ due to the updates caused. Clearly, $\hat{z} \in \Delta A$. Add all edges incident to $\Delta A$ in $\hat{D}_j$ to $\hat{H}_i$ (if $e \in F_i$, replace $\hat{u} = \hat{v}$ by the new vertices $\hat{u}$ and $\hat{v}$ resulting from the vertex split and move the edges incident to $\hat{u}$ to the right endpoint). Additionally, add $2 \cdot \gamma_{prune} \cdot \kappa/\phi_{\textit{sparse}}$ edges from $\hat{D}_j \setminus \hat{H}_i$ to $\hat{H}_i$ (as long as such edges exist).  
        
        Then, induce the tree $\hat{T}_j$ and the graph $\hat{D}_j$ on vertex set $V_{\hat{D}_j} \setminus \Delta A$. After inducing, the tree $\hat{T}_j$ is potentially a forest. To remedy this, while there are components $C_1, C_2, \ldots, C_\nu$ for $\nu \geq 2$, take the smallest such connected component, say $C_1$, and sample uniformly at random an edge incident to $C_1$ in $\hat{D}_j$ until an edge $\hat{f}$ is found that leaves $C_1$. Add the edge $\hat{f}$ to $\hat{T}_j$ and $\hat{H}_j$ and repeat the process.\footnote{If re-connecting takes time $\gg O(\log n)  \cdot \gamma_{\textit{prune}}/\phi_{\textit{sparse}}$, we re-start the whole data structure. We show that this only happens with probability $1/n^C$ for some large constant $C$.}

        Finally, add the new projections $\hat{u}^{\textit{NEW}}, \hat{v}^{\textit{NEW}}$ back into $\hat{H}_i$ together with their incident edges (now potentially with different endpoints). 
    \end{itemize}

    Then, we handle the set of edge insertions to $H_{i-1} \cup F_i$:
    \begin{itemize}
        \item For every edge $e$ inserted, if $\hat{e} = \bot$, we do nothing; otherwise we add $\hat{e}$ to $\hat{H}_i$.
    \end{itemize}
\end{itemize}

\paragraph{Analysis.} Again, let us analyze the procedures.

\begin{claim}
At all times $\hat{H}_i$ remains a connectivity sparsifier of $\hat{G}_i$. 
\end{claim}
\begin{proof}
For the re-initialization procedure, correctness follows immediately since spanning trees are clearly connectivity sparsifiers of connected graphs.  

In our update procedure, we have that insertions to $H_{i-1} \cup F_i$ are added to $\hat{H}_i$ if they are added to $\hat{G}_i$, and thus connectivity is preserved.

Let us now analyze the deletion of an edge $(u,v) \in H_{i-1} \cup F_i$: we first delete the images of endpoints $\hat{u}^{\textit{OLD}}, \hat{v}^{\textit{OLD}}$ with all their incident edges from $\hat{H}_i$. Then, we prune some vertices from graphs $\hat{D}_j$, add the edges incident to the pruned vertex set to $\hat{H}_i$ and then add the (fixed) tree $\hat{T}_j$ to $\hat{H}_i$ for each affected graph. 

But this implies that $\hat{H}_i$ consists of spanning trees of connected components formed by each $\hat{D}_j$ and all remaining edges in $\hat{G}_i$ by the end of the update. Thus, clearly, $\hat{H}_i$ is a connectivity sparsifier of $\hat{H}_i$ at any time.
\end{proof}

\begin{claim} \label{clm:size}
At each time $t$ divisible by $x_i$, graph $H_i$ consists of at most $\frac{108}{\kappa} \cdot |E_{H_{i - 1}}| + 2 \gamma_{\textit{exp}} \cdot \phi_{\textit{sparse}} \cdot |E_{H_{i-1}}|$ edges. At each time $t$ not divisible by $x_i$, $H_i$ undergoes at most one edge deletion and at most $4i \cdot \kappa^2 \cdot \gamma_{\textit{prune}}/\phi_{\textit{sparse}}^2$ edge insertions.
\end{claim}
\begin{proof}
    At each time $t$ divisible by $x_i$, layer $i$ is rebuilt. Let us first bound the number of edges it contains after a rebuild. 
    
    By \Cref{clm:correctnessKappaShattering}, we have $|E_{\hat{A}_i}| \leq 5 \cdot |E_{H_{i - 1}}|$. We directly obtain $|E_{\hat{G}_i}| \leq |E_{\hat{A}_i}| \leq 5 \cdot |E_{H_{i - 1}}|$. We first bound the number of edges between expander components $X_j$ in the expander decomposition. By \Cref{thm:expander_decomp}, the total number of such edges is bounded by $\phi_{\textit{sparse}} |E_{\hat{G}_i}| \leq 5 \cdot \gamma_{\textit{exp}} \cdot \phi_{\textit{sparse}} \cdot |E_{H_{i - 1}}|$. All other edges in $\hat{H}_i$ are part of spanning trees $\hat{T}_j$ of expander components $X_j$. We notice that singleton compoents do not contribute any edges. Therefore, the total number of such edges is at most $\frac{108}{\kappa}|E_{H_{i - 1}}|$. This concludes the proof of the first part of the claim. 
    
    If the update at time $t$ is an insertion, then $H_i$ undergoes at most one insertion and no deletions. Therefore, we focus on edge deletions. Since $H_i \subseteq G$, a deletion to $G$ clearly only causes one deletion to $H_i$. 

    Recall that after every deletion, each layer $j \leq i$ outputs a set of edges $\Delta \hat{E}_j$, and $H_i$ is updated by adding the edge set $\bigcup_{j \leq i} \Delta E_j$. We first bound the size of the set $\Delta \hat{E}_j$ for every $j$. If both mapped endpoints of the edge are not contained in some graph $\hat{D}_k$ at this layer, the set is empty. Otherwise, let us assume that the data structure $\mathcal{P}_k$ has not yet processed the maximum number of updates. In that case simulating the edge deletion/vertex split on $\hat{D}_k$ takes at most $\kappa/3$ edge deletions to $\mathcal{P}_k$. Each of these will cause up to $\kappa \gamma_{\textit{prune}}/\phi_{\textit{sparse}}^2$ edges to enter the set $\Delta E_j$ by \Cref{thm:worst_case_pruning}. Finally another $\kappa 2\gamma_{\textit{prune}}/\phi_{\textit{sparse}}$ arbitrary edges are added. Afterwards, the tree $\hat{T}_k$ is repaired. Since there are at most $\kappa \gamma_{\textit{prune}}/\phi_{\textit{sparse}}^2$ edge deletions to $\hat{T}_k$, this adds at most $\kappa \gamma_{\textit{prune}}/\phi_{\textit{sparse}}^2$ extra edges. 
    
    Finally, we observe that whenever $\mathcal{P}_k$ processed the maximum amount of deletions, the extra edges added in each previous step ensure that every edge incident to $D_k$ is already in the graph $H_k$. 

    In this proof, we ignored failures. Whenever such a failure occurs with probability $1/n^C$ for some large constant $C$, we re-start the whole algorithm and therefore $t$ is re-set to $0$. 
\end{proof}

\begin{claim} \label{clm:rebuild}
The re-initialization of $\hat{H}_i$ and $H_i$ can be implemented to run in time $\tilde{O}(|E_{H_{i-1}}|)$ with high probability. For every other update, the algorithm requires expected time $\tilde{O}(1)$ to update $\hat{H}_i$ and $H_i$.
\end{claim}
\begin{proof}
    We first compute the tree $F_i$ in time $\tilde{O}(|E_{H_{i-1}}|)$ by \Cref{clm:forsest_rt}. Then, we compute an expander decomposition of the resulting core graph $\hat{A}_i$. This graph contains at most $\tilde{O}(|E_{H_{i-1}}|)$ edges by \Cref{clm:correctnessKappaShattering}. Therefore, the runtime of computing said expander decomposition follows from $\phi_{\textit{sparse}} = \tilde{\Omega}(1)$ and \Cref{thm:expander_decomp}. Finally, we initialize the pruning data structures on the expander components, which takes total time $\tilde{O}(|E_{H_{i-1}}|)$ by \Cref{thm:worst_case_pruning}. This concludes the runtime analysis of the re-initialization. 

    We analyze the time spend with pruning and recomputing trees separately.
    \begin{itemize}
        \item The pruning data structure takes total time $\tilde{O}(1)$ to prune at most $\kappa$ edges by \Cref{thm:worst_case_pruning} and the definition of our algorithm, where we use that $\kappa = \tilde{O}(1)$.
        \item The sampling procedure succeeds with probability $\gamma_{\textit{prune}}/\phi_{\textit{sparse}} = \tilde{\Omega}(1)$. Therefore, the expected update time follows directly from the fact that the total number of new edges added to the tree is at most $4\kappa^2 \cdot \gamma_{\textit{prune}}/\phi_{\textit{sparse}}^2 = \tilde{O}(1)$.
    \end{itemize}
    All other operations can be implemented in time $\tilde{O}(1)$.
\end{proof}

\subsection{Analyzing the Main Algorithm} \label{sec:analysis_main}

To analyze the correctness of our algorithm, we first prove the main invariant. 

\begin{proof}[Proof of \Cref{inv:keyInv}] 
    We first observe that for $i = 0$, we have $F_0 = (V, \emptyset)$ for every $t$, and therefore the number of connected components in $F_0[C]$ is $|C| = |C|/2^0$ as desired. 
    
    We then strengthen the invariant for $i > 0$ to enable our proof. 
    \begin{itemize}
        \item \underline{Strengthened Invariant:} For every layer $i \geq 0$ and time $t$, we have
        \begin{equation*}
        \text{(\#connected components in $F_i$)} \leq \text{(\#connected components in $G$)} + 2^{\Lambda - i - 2} + (t - t_i).
    \end{equation*}
    where $t_i$ denotes the largest timestamp such that $t_i \leq t$ and $x_i$ is divisible by $t_i$. 
    \end{itemize}
    We prove the strengthened invariant by induction on $i$ and $t$.
    
    For $t = 0$ and $i = 0$ we have $F_0 = (V, \emptyset)$, and thus the number of connected components in $F_0$ is $|V| \leq 2^{\Lambda - 2}$ as desired. Then, we fix some timestamp $t$ and layer $i$, and assume the strengthened invariant holds for all pairs layers $j$ and timesteps $t'$ such that either $j < i$ and $t' \leq t$ or $t' < t$. There are two cases to consider. 
    \begin{itemize}
        \item \underline{$x_i$ is divisible by $t$:} In this case, the layer $i$ gets rebuilt. By the induction hypothesis, we have that $F_{i - 1}$ contains at most $2^{\Lambda - i - 1}$ more components than $G$ because $t - t_{i - 1} \leq 2^{\Lambda - i - 2}$ throughout by the definition of our algorithm. Therefore, the claim follows from \Cref{clm:correctnessKappaShattering} since $\kappa = 8 \cdot 108$. 
        \item \underline{Otherwise:} In this case, we invoke the induction hypothesis for $t - 1$ and $i$.
        
        If the update is an edge deletion, then this can cause the number of connected components in $F_i$ to go up by at most one. Therefore the strengthened invariant follows directly. 

        If the update is an edge insertion, we distinguish two cases. Either the endpoints were already in the same connected component at time $t - 1$. In this case, the strengthened invariant follows by induction since the forest $F_i$ remains unchanged. Or the edge insertions connects two previously disconnected components $C_1$ and $C_2$. Then we have that the total number of connected components of $F_i$ remains the same, but the number of connected components in $G$ decreases by exactly one. Therefore, the invariant again follows. 
    \end{itemize}
    We finally observe that $t - t_i \leq 2^{\Lambda - i - 3}$. Therefore, the strengthened invariant implies the desired invariant. 
\end{proof}

\begin{lemma}\label{lma:update_time}
    The amortized update time is $\tilde{O}(1)$. 
\end{lemma}
\begin{proof}
    We first prove that $|E_{H_{i}}| = \tilde{O}(x_i)$ by induction. For our proof, we again work with a strengthened invariant. 
    \begin{itemize}
        \item \underline{Strengthened Invariant:} $|E_{H_i}| \leq 4 i \cdot \kappa^2 \cdot \gamma_{\textit{prune}}/\phi_{\textit{sparse}}^2 \left( 2^{\Lambda - i} + (t - t_i) \right) $, where $t_i$ denotes the largest timestamp such that $t_i \leq t$ and $x_i$ is divisible by $t_i$. 
    \end{itemize}
    For $i = 0$, the claim follows for every $t \leq m$ since $|E_{H_0}| \leq m \leq 2^\Lambda$ where we recall that $x_0 > m$ and therefore the layer $0$ never gets rebuilt. 
    Then, we again fix some timestamp $t$ and layer $i$, and assume the strengthened invariant holds for all pairs layers $j$ and timestamps $t'$ such that either $j < i$ and $t' \leq t$ or $t' < t$. 
    \begin{itemize}
        \item \underline{$x_i$ is divisible by $t$:} In this case, the layer $i$ gets rebuilt. By the induction hypothesis, we have that $H_{i - 1}$ contains at most $4(i - 1) \cdot \kappa^2 \cdot \gamma_{\textit{prune}}/\phi_{\textit{sparse}}^2 (2^{\Lambda - i - 1} + 2^{\Lambda - i - 2})$ edges because $t - t_{i - 1} \leq 2^{\Lambda - i - 2}$ throughout by the definition of our algorithm. Therefore, the invariant follows from \Cref{clm:size} since $\kappa = 8 \cdot 108$ and $\phi_{\textit{sparse}} = 1/(20 \gamma_{\textit{exp}})$.
        \item \underline{Otherwise:} In this case, we invoke the induction hypothesis for $t - 1$ and $i$. If the update is an edge deletion, then this can cause the number of connected components in $F_i$ to go up by at most one. Then, the invariant directly follows from \Cref{clm:size}. 
    \end{itemize}
    Finally, we observe that the invariant implies $|E_{H_i}| = \tilde{O}(x_i)$ by our choice of parameters and \Cref{thm:worst_case_pruning}. Therefore, the claim directly follows because layer $i$ gets re-initialized after $x_i$ time steps, and the cost of re-initialization of layer $i$ is $\tilde{O}(H_{i - 1}) = \tilde{O}(x_i)$ \Cref{lma:update_time}. All other updates take worst-case time $\tilde{O}(1)$ by \Cref{lma:update_time}, so the total update time follows by summing over the $\Lambda = O(\log m)$ layers. 

    Finally, failures happen with sufficiently low probability to amortize away a full rebuild. 
\end{proof}

\mainThm*
\begin{proof}
    Follows directly from \Cref{inv:keyInv} and \Cref{lma:update_time}. 
\end{proof}

\subsection{De-amortization} \label{sec:deamortization}

As presented, our algorithm only achieves amortized update times. This deficit is remedied via standard de-amortization techniques. 

We observe that the only amortized part in the update time bound is the re-initialization, which happens in regular intervals. This allows pre-computing such re-initializations in the background and swapping them in at the opportune moment to obtain an algorithm with worst-case update times. 

We remark that this technique heavily relies on pointer switching. Concretely, the algorithm cannot output changes to the forest explicitly, and instead provide query access to the forest. However, in \Cref{sec:explicit}, we provide a black-box reduction that shows that such an algorithm suffices to maintain an explicit tree. 

\subsection{Maintaining an Explicit Forest} \label{sec:explicit}

In this section, we provide a simple reduction from explicit to implicit forest maintenance. 

\begin{lemma} \label{lma:explicit_repair_tree}
    Assume to be given a worst-case $\tilde{O}(1)$ update time dynamic connectivity algorithm that provides query access to a maximal spanning forest. Then, there is a worst-case $\tilde{O}(1)$ update time dynamic connectivity algorithm that explicitly maintains a maximum spanning forest via at most $2$ updates per step. 
\end{lemma}
\begin{proof}
    The algorithm initially computes an explicit maximal spanning forest $F$. In the following, we let $F'$ denote the implicitly maintained forest which we query. We assume that both are stored in link-cut tree data structures \cite{driscoll1986making}. 
    
    Whenever an edge is inserted, the algorithm first checks if it crosses tree components in $F$. If so, it adds it to $F$. 

    Whenever an edge $(u,v)$ is deleted from $F$, the algorithm checks if $u$ and $v$ are in the same component in $F'$. If so, it needs to repair the tree $F$. To do so, it sets $x \gets u$, $y \gets v$. Then, it uses the link-cut tree to obtain the midpoint $z$ of the path $F'[x, y]$. It then checks if $z$ is in the same components as $u$ in $F$. If so, we set $x \gets z$. Otherwise, we set $y \gets z$. We continue until the edge $(x,y)$ is in $F'$. Then, we add this edge to $F$.

    This procedure terminates in $\tilde{O}(1)$ steps since the length of the path is halved in every step. Furthermore, $x$ is always in the same component of $F$ as $u$, and $y$ is always in the same component as $v$. Therefore, the edge $(x,y)$ repairs the tree in $F$ as desired. 

    Other updates do not affect the connected components of the graph.
\end{proof}

Similarly to \Cref{lma:explicit_repair_tree}, we can also obtain an algorithm that maintains a $(1 + \epsilon)$ approximate minimum spanning tree for a graph with edge weights in $[1, U]$.

\begin{corollary}\label{corr:apx_mst}
    Assume to be given a worst-case $\tilde{O}(1)$ update time dynamic connectivity algorithm that provides query access to a maximal spanning forest. Then, there is a worst-case $\tilde{O}(\log(U)/\epsilon)$ update time algorithm that explicitly maintains a $(1 + \epsilon)$ minimum spanning forest of a weighted graph $G$ with edge weights in $[1, U]$ via at most $2$ updates per step.  
\end{corollary}
\begin{proof}
    For $i = 0, \ldots, k = O(\log(U/\epsilon))$, let $G_i$ be the subgraph of $G$ that only contains edges with weight at most $(1 + \epsilon)^i$.

    We then maintain an explicit spanning forest for each of the graphs $G_i$. We observe that the connected components of $G_i$ contain the connected components for $G_j$ whenever $j < i$. We say an edge has level $i$ if $G_i$ is the first graph in which it appears. 

    We then maintain a minimum level spanning tree, i.e. a spanning tree of minimum total level. Whenever an edge $(u,v)$ of level $i$ in that tree gets deleted, we first check if the corresponding spanning forest got repaired. Otherwise, we search for the lowest level $j > i$ such that $u$ and $v$ are in the same forest and locate a repairing edge via binary search as in the proof of \Cref{lma:explicit_repair_tree}. 

    When a new edge is inserted, we search for the highest level edge in the cycle formed by the maintained approximate minimum spanning tree and the edge via the same technique. If this level is higher than the level of the newly inserted edge, then we replace it in the tree. 

    This procedure clearly maintains a minimum level spanning tree and can be implemented to run in time $\tilde{O}(k)$. 

    The claim follows since for every edge $e$ in the tree, there is no edge of with weight $< w(e)/(1 + \epsilon)$ that connects the tree components after deleting the edge $e$. 
\end{proof}

\subsection{A path towards De-randomization} \label{sec:derandomziation}

The connectivity sparsifier is the only randomized part of our algorithm. There, randomness is used to quickly compute an expander decomposition, and to maintain a spanning tree of an expander graph. 

The broadest path towards de-randomizing this algorithm is to provide a connectivity sparsifier $H \subseteq G$ of $G$ with an explicit embedding $\Pi_{G \mapsto H}$ from $G$ to $H$ (see \Cref{def:graphEmbedding}). Such an embedding maps each edge $e = (u, v)$ that is present in $G$ to a $uv$-path $\Pi_{G \mapsto H}(e)$ in $H$, and ensures that every edge in $H$ appears on at most $\alpha \cdot |E_G|/|E_H|$ such paths for some parameter $\alpha$. The ability to \emph{statically} compute such a sparsifier suffices, because it can be maintained by simply adding all edges whose embedding paths are broken to the sparsifier. More precisely, given an edge $e$ from $G$ is deleted, then we add to $H$ all edges $e'$ where $e \in \Pi_{G \mapsto H}(e')$. 

The algorithm only adds at most $\alpha \cdot |E_G|/|E_H|$ edges to $H$ by definition of $\alpha$. Since we ensure for each graph $\hat{G}_i$ in the hierarchy that $|E(\hat{G}_i)| = \tilde{O}(|V(\hat{G}_i)|)$, i.e. the number of edges exceeds the number of vertices by at most a polylogarithmic factor, we only add for each level $\tilde{O}(\alpha)$ edges to the hierarchy across all levels.

Further, $H$ preserves the connectivity information of $G$ at all times: for every edge $e'$ in $G$, either there still is an explicit path between its endpoints in $H$, or the edge $e'$ is itself in $H$. Thus, it is simple to show by induction on the levels that the connectivity algorithm remains correct.

Via the same analysis as in previous sections, we derive that the new algorithm runs in worst-case update time $\tilde{O}(\alpha + T(m))$ where $T(m) \cdot m$ is the runtime of the algorithm that statically computes sparsifier and embedding. For $\alpha = O(T(m))$, this yields \Cref{thm:reduction}.

A sparsifier as described above along with a low-congest embedding can be obtained rather directly via edge-decremental APSP \cite{kyng2023dynamic, kyng2024bootstrapping} using ideas similar to \cite{maxflow, chen2023simple}, but the total time incurred by these data structures is currently $m^{1+o(1)}$ instead of $\tilde{O}(m)$ for any $m$ edge graph. This yields an immediate de-randomization of our algorithm with subpolynomial worst-case update time.

Another plausible way toward deterministic algorithms is to de-randomize the components of our algorithm more directly. A fundamental stepping stone would be the development of a deterministic near-linear time algorithm to compute expander decompositions. The current state-of-the-art algorithm takes time $m^{1+o(1)}$. In developing such an algorithm, one can also hope to derive techniques that allow us to route more efficiently in expander graphs. Most concretely, interpreting each edge in $G$ as a commodity that has to be embedded into $H$, computing an embedding boils down to solving a $k$-commodity routing problem for $k = \tilde{\Theta}(n)$. Recent techniques \cite{haeupler2024low} can solve such problems approximately in time $\tilde{O}( (m+k)^{1+\epsilon})$ for arbitrarily small $\epsilon > 0$. 

\section*{Acknowledgments}

We thank Yu-Cheng Yeh and the anonymous reviewers for comments that improved the presentation of this article.

\newpage
\bibliographystyle{alpha}
\bibliography{refs}

\end{document}